\documentclass{elsarticle}
\usepackage {graphicx}
\usepackage {amsmath}
\usepackage {amssymb}
\usepackage {amsthm}
\newtheorem{lem}{Lemma}
\begin{document}
\begin{frontmatter}
\title{Flipping the Winner of a Poset Game}

\author{Adam O. Kalinich}
\ead{akalinich@imsa.edu}
\address{Illinois Math and Science Academy, 1500 Sullivan Road, Aurora, Illinois 60506}

\begin{abstract}
Partially-ordered set games, also called poset games, are a class of two-player combinatorial games. The playing field consists of a set of elements, some of which are greater than other elements. Two players take turns removing an element and all elements greater than it, and whoever takes the last element wins. Examples of poset games include Nim and Chomp. We investigate the complexity of computing which player of a poset game has a winning strategy. We give an inductive procedure that modifies poset games to change the nim-value which informally captures the winning strategies in the game. For a generic poset game $G$, we describe an efficient method for constructing a game $\neg G$ such that the first player has a winning strategy if and only if the second player has a winning strategy on $G$. This solves the long-standing problem of whether this construction can be done efficiently. This construction also allows us to reduce the class of Boolean formulas to poset games, establishing a lower bound on the complexity of poset games. 
\end{abstract}
\end{frontmatter}


\section{Introduction}

\subsection{Definition of a Poset Game}

A \textbf{p}artially-\textbf{o}rdered \textbf{set} consists of a set $V$ of elements $\{v_1,v_2,v_3,...\}$, and an ordering relation $\ge$.  This ordering relation is reflexive, transitive, and anti-symmetric. It is possible that for some elements, neither $v_i \ge v_j$ nor $v_j \ge v_i$ holds. A poset game is a two-player game played on a poset. A move in a poset game consists of choosing an element $v_i$ and removing all elements $v_j$ such that $v_j \ge v_i$, so that a smaller poset remains. The two players alternate moves. If a player cannot move because all elements have been removed, that player loses.

\subsection{Impartial Games and Nim-Values}


Many different kinds of Poset games have been studied by mathematicians. The oldest such game is Nim, the winning strategy of which was discovered in 1902 \cite{bouton1901nim}. A game of Nim is played on several piles of elements, and a move consists in taking any number of elements from any one pile. The player who cannot move loses. Nim was solved by finding a simple method to determine if any position was winning or losing. The nim-sum, denoted $\oplus$, of two numbers is computed by representing the numbers in binary and taking the bitwise parity of the two. In nim, the nim-value of a position is the nim-sum of the number of elements in each pile. If a position does not have nim-value 0, elements can be removed from a pile so that the nim-value is 0. The next player must change the nim-value to something else, but then the other player can just change it back to 0. Therefore, the player who first moves from a position of a nim-value of 0 can be forced to move from every position of nim-value 0 reached in the game. The empty position without any elements left has nim-value 0, so the player who first moved from a position of nim-value 0 can be forced to lose. If the nim-sum of the number of elements in each of the piles is 0,  then the game is winning, and if it is not 0, then it is losing. See \cite{bouton1901nim} for a more detailed explanation and proof of the winning strategy.

A pile of $x$ elements in Nim can be seen as an $x$-tower in a poset game, where an $x$-tower is a tower of $x$ vertices with $v_1 \le v_2 \le v_3 \le ... \le v_x$. Since any game of Nim can be represented like this, Nim is actually a specific type of poset game. Other types of poset games include Schuh's game of divisors\cite{Schuh52}, proposed in 1952 by F. Schuh, played with a number and all its divisors expect 1, where a move consists of taking away a divisor and each of its multiples. Another example is Chomp, proposed in 1974 by D. Gale \cite{Gale74}, played on a rectangular grid with the upper left square missing, where a move consists of taking away a square and all squares below and to the right of it.

\subsection{Poset Games}

The Sprague-Grundy theorem, discovered independently by Sprague \cite{Sprague35} and Grundy \cite{Grundy39} states that impartial games where the last player to move wins are equivalent to a pile in nim. The nim-value $Nim(G)$ of any impartial game $G$ is defined recursively as the smallest non-negative integer for which there does not exist a move to a game of that nim-value. Since only finite poset games will be considered, and poset games are impartial and have last-player-win rules, the Sprague-Grundy theorem can be applied. By the Sprague-Grundy theorem, if two poset games, $G_1$ and $G_2$, are put next to each other to make a new game, then the resulting game has a nim-value of $Nim(G_1) \oplus Nim(G_2)$. The second player has a winning strategy on a poset game $G$ if and only if $Nim(G)=0$






\section{Constructing $\neg G$}

\subsection{Overview of the Construction}

	Let Poset game $G$ have size $g$. We will describe a method for constructing a game $\neg G$ such that the first player has a winning strategy if and only if the second player has a winning strategy on $G$. $Nim(G)$ does not need to be computed to construct $\neg G$. We will add vertices to construct a game $G'$ such that $Nim(G')$, when expressed in binary, has a $1$ in a known digit if and only if $Nim(G) \ne 0$. We will use this to construct a game $G'''$ with a $Nim(G''')$ equal to $2^{\lfloor {Log_2(g)} \rfloor+3}$ if $G$ was winning, and 0 otherwise. By putting this next to a game with nim-value $2^{\lfloor {Log_2(g)} \rfloor+3}$, we have constructed a game $G''''$ in which the first player has a winning strategy if and only if the second player has a winning strategy on $G$.

\subsection{Lemmas Necessary for the Construction}

A bottom vertex is a vertex that is less than all other vertices in the game. The process of adding a bottom vertex consists of adding a vertex that is defined as less than all other vertices at the time it is added. If more vertices are added, they may be defined as less than, greater than, or incomparable to this vertex.

\begin{lem}
Adding a bottom vertex increases the nim-value of a poset game by 1.
\end{lem}
\begin{proof}
We will use induction on j, the size of the game. For the base case, a poset game of size 1 has a nim-value of 1, and when a bottom vertex is added, the nim-value is 2. For a poset game of size $j$ with a nim-value of $k$, there exist moves to games of nim-values 0,1,2,3,...,$k-1$, but no move to a game of nim-value $k$ exists. By the induction hypothesis, the nim-values of these games change to 1,2,3,4,...,$k$ when a bottom vertex is added, since these games all have fewer vertices than the original game. A move to a zero game exists, by taking away the bottom vertex. No move to a game of nim-value k+1 exists, since no game with nim-value $k$ existed before the bottom vertex was added. Since there are moves to games of nim-values 0,1,2,3,…$k$, but not $k+1$, the nim-value has changed from $k$ to $k+1$ by adding a bottom vertex. This proves the lemma.
\end{proof}

\begin{lem}
An $x$-tower has a nim-value of $x$.
\end{lem}

\begin{proof} 
We will use induction on k. A $1$-tower has a nim-value of $1$. From a $k$-tower there exist moves to towers of every size from $0$ to $k-1$, and by the induction hypothesis those games have nim-values of every value from $0$ to $k-1$. Thus, the $k$-tower has nim-value $k$, proving the lemma.
\end{proof}

We have two operations for modifying nim-values in a predictable way. When we add a bottom vertex, we increase the nim-value by one, so adding $x$ bottom vertices in sequence to a a game $G$ will result in a game with nim-value of $Nim(G)+x$.  By putting an $x$-tower next to a game $G$, we get a game with a nim-value of $Nim(G) \oplus Nim(x$-tower$)=Nim(G) \oplus x$. By using the operations $\oplus x$ and $+x$ in pairs, we can modify positive nim-values without modifying the nim-value of games with nim-value 0. We will first describe how the process affects games with non-zero nim-values, and then we will show how it affects games with nim-value 0.

\subsection{Constructing $\neg G$}

First, we will look at the case where the nim-value is not zero, and show a construction that flips the winner. After, we will show that this construction also flips the winner if the nim-value is zero.
\newtheorem{thm}{Theorem}
\begin{thm}For any poset game $G$ with g vertices, if $Nim(G) \neq 0$, the following procedure will create a game $\neg G$ with $Nim(\neg G)=0$:

1. Set $i=0$

2. While $2^i \leq g$, add $2^i$ bottom vertices, then put it next to a $2^i$-tower, then increment $i$ by 1.

3. Let the resulting graph be $G'$

4. Add $2^i$ bottom vertices, and put it next to a $2^i$-tower to get $G''$. Increment $i$ by 1.

5. Put $G'$ next to $G''$ and a $2^i$ tower to get $\neg G$
\end{thm}

\bigskip

\begin{proof}
	The first step of the construction of $\neg G$ is to add a bottom vertex to G and then to put a 1-tower (a single vertex) on the side. Then, add 2 bottom vertices and then a 2-tower to the side. We will keep repeating this process for powers of 2. Since we are using powers of 2, this only acts on a single digit at a time, although carrying that results from the addition could modify digits further to the left. If we continue this process, then there will always be a binary digit \textquotedblleft1" to the left of the last digit modified directly. We will show this by examining the step in which we add $2^k$ and then nim-add $2^k$ to a nim-value $Nim(G)$. $2^k$ has only one digit in binary, so it is easy to add and nim-add with it. 
\begin{center}
Case 1:
\begin{align*}
a&0b\\
+&100000...\\
\oplus&100000...\\
=a&0b\\
\end{align*}
\end{center}
So if the $2^k$th digit of $Nim(G)$ is \textquotedblleft0", then there is no change to $Nim(G)$. The \textquotedblleft1" to the left has not been changed and there is still a \textquotedblleft1" to the left.
\begin{center}
Case 2:
\begin{align*}
a0&1b\\
+&100000...\\
\oplus&100000...\\
=a1&1b\\
\end{align*}
\end{center}
If the $2^k$th digit of $Nim(G)$ is \textquotedblleft1" and there is a \textquotedblleft0" to its immediate left, then that \textquotedblleft0" gets changed to a \textquotedblleft1". A new \textquotedblleft1" has been generated to the left.
\begin{center}
Case 3:
\begin{align*}
a0111...111&1b\\
+&100000...\\
\oplus&100000...\\
=a1000...000&1b\\
\end{align*}
\end{center}
If the $2^k$th digit of $Nim(G)$ is \textquotedblleft1" and there is a string of \textquotedblleft1"s of any non-zero length to its immediate left, then that string of \textquotedblleft1"s is changed to a string of \textquotedblleft0"s and the \textquotedblleft0" to the left of the string is changed to a \textquotedblleft1". A new \textquotedblleft1" has been generated to the left. 

Because there are only $g$ possible moves, $ Nim(G) \le g$. Therefore, when $Nim(G)$ is expressed in binary notation, it has at most $\lfloor {Log_2(g)} \rfloor+1$ digits. Let the continuation of this process for $\lfloor {Log_2(g)} \rfloor+2$ total steps result in game $G'$ with nim-value Nim(G'). We will have acted on every digit in the original nim-value, with $2^{\lfloor {Log_2(g)} \rfloor+1}$ being the last digit that was directly acted on. There still will be a "1" to the left of the last digit modified directly. This digit could not have been part of the original nim-value, and since a new "1" can only be created to the immediate left of a "1", there must have been a "1" in the $2^{\lfloor {Log_2(g)} \rfloor+2}$ place during some point of the process. If that "1" was created by case 3, then only "0"s would be acted on from that point on, so the nim-value would not change anymore. If that "1" was created by case 2, then it was created in the last step, and thus the nim-value did not change after that. So, there is a "1" in the $2^{\lfloor {Log_2(g)} \rfloor+2}$ digit of $Nim(G')$ place.

We then construct another copy of $G'$ using the same process. With this copy of $G'$, add $2^{\lfloor {Log_2(g)} \rfloor+2}$ bottom verticies and then put a $2^{\lfloor {Log_2(G')} \rfloor+2}$-tower on the side, and let the resulting game be $G''$. Then, we put $G'$ next to $G''$ to make game $G'''$.
\begin{center}
\begin{align*}
Nim(G')=&01c\\
Nim(G'')=&11c\\
Nim(G''')=Nim(G') \oplus Nim(G'')=&1000...=2^{\lfloor {Log_2(g)} \rfloor+3}
\end{align*}
\end{center}
$Nim(G'')$ only differs from $Nim(G')$ in a single, known digit, so when they are nim-added, the resulting game $Nim(G''')$ has a nim-value of $2^{\lfloor {Log_2(g)} \rfloor+3}$. The final step is to put a $2^{\lfloor {Log_2(g)} \rfloor+3}$-tower next to $G'''$, resulting in game $G''''$ that has nim-value 0 if the original game was winning.
\end{proof}

Now, we will look at how the same exact process would affect a game with nim-value zero. 

\begin{thm}
For any poset game $G$, if $Nim(G) = 0$, the procedure described earlier will create a game $\neg G$ with $Nim(\neg G) \neq 0$
\end{thm}

\begin{proof}
Since adding a number to 0 and then nim-adding the same number cancel each other out, nearly every step had no effect on the nim-value.
\begin{center}
\begin{align*}
Nim(G'')&=Nim(G')=Nim(G)=0 \\
Nim(G''')&=Nim(G') \oplus Nim(G'')=0 \oplus 0=0\\
Nim(G'''')&=Nim(G''')+2^{\lfloor {Log_2(g)} \rfloor+3}=2^{\lfloor {Log_2(g)} \rfloor+3}
\end{align*}
\end{center}
We have found that if $Nim(G)$ was positive, $Nim(G'''')$ is 0, and if $Nim(G)$ was 0, $Nim(G'''')$ is positive. The first player had a winning strategy for $G''''$ if he did not have a winning strategy for $G$, and he will not have a winning strategy for $G''''$ if he did have a winning strategy for $G$. Therefore, $G''''$=$\neg G$.
\end{proof}

The relation between any two vertices in the graph can easily be computed in polylogarithmic time given the edges of the old graph. If one vertex was added as a bottom vertex, it can be checked if the other vertex was added before it or not. If a vertex was added as part of a side tower, it can easily be checked whether the other vertex is higher in the tower. If a vertex was in the original grame, then the other vertex can only be greater than it if it was greater in the original game.

\begin{lem}
The size of $\neg G$ is only a linear blowup from $G$
\end{lem}

\begin{proof}
\begin{center}
\begin{align*}
Size(G')=&Size(G)+2*\displaystyle\sum\limits_{i=0}^{\lfloor {Log_2(g)} \rfloor+1} 2^i\\
=&g+2^{\lfloor {Log_2(g)} \rfloor+2}-2 \leq 5g\\
Size(G'')=&Size(G')+2*2^{\lfloor {Log_2(g)} \rfloor+2}\\
\leq&5g+8g=13g\\
Size(G''')=&Size(G')+Size(G'')\\
\leq&5g+13g=18g\\
Size(\neg G)=&Size(G''')+2^{\lfloor {Log_2(g)} \rfloor+3}\\
\leq&18g+8g=26g \\ \\
\end{align*}
\end{center}
This shows that it is possible to construct the $not$ of a game with only a linear blowup. 
\end{proof}

\section{Reduction of Boolean Formulas to Poset Games}

\subsection{Constructing OR and AND Gates}

We will now show how to construct $A \cup B$ of two games $A$ and $B$ with only a linear blowup. Computing the OR of poset games is a folklore result, but we give the construction and proof for completeness. Given game $A$ with elements {$a_1,a_2,a_3,..., a_v$} and game $B$ with elements {$b_1,b_2,b_3,..., b_w$}, game $A \cup B$ has elements {$a_1,a_2,a_3,...,a_v,b_1,b_2,b_3,...,b_w$} and the additional set of relations $b_i \ge a_j$


If the first player can win $A$, he plays the winning move in game $A$. $B$ disappears, so only the remainder of game $A$ is being played, and since the first player just played the winning move there, he will win. If the first player cannot win game $A$, but he can win game $B$, he can play the winning move in game $B$. Both players would try to avoid playing in game $A$, since any move there is losing. So, they both play in game $B$, trying to force the other to make the first move in $A$. Since the first player can win game $B$, he can force the second player to play first in game $A$, so will be able to win. If both game $A$ and game $B$ are losing, then the second player can force the first player to make the first move in $A$, so the first player will not be able to win. We see that the first player will win this game if and only if he can win game $A$ or game $B$. 

It is possible to construct an OR of two games and a NOT of a game with only a linear blowup, and since $\neg ( \neg A \cup \neg B )$=$A \cap B$, we can construct the AND of two games with only a linear blowup by using ORs and NOTs. Before our results, there was no known efficient procedure for computing the AND of poset games.

\subsection{Reducing $NC^1$ circuits to Poset Games}

Given a boolean formula, we can represent false variables with a poset game of two isolated vertices, a game which is losing for the first player. We can represent true variables with a single-vertex game, which is winning for the first player. By using the constructions we have shown, we can model Boolean circuits with bounded fan-in. The result of a boolean formula will be a single poset game, which will be a first-player win if and only if the formula evaluates to true. If the boolean formula is one that can be evaluated by log-depth circuits, the poset game it is reduced to will have polynomial-size. This shows that finding the winning player of a poset game is an $NC^1$-hard problem.


\section{Future Work}
The gap between $NC^1$-hard and PSPACE is very large, and closing the gap is a possible topic for future work. The method of modifying poset games used in this paper might be used to reduce to and from other problems. Results by Byrnes \cite{Byrnes03} suggest that poset games may not be PSPACE-complete, and could in fact be far easier. Expanding periodicity theorems to multiple rows of chomp might lead to a way to calculate the nim-value of a generic poset game.


\nocite{Soltys10}

\nocite{WinningWays}

\section{Acknowledgements}
This work would not have happened without the guidance and supervision of Dr. Lance Fortnow. He introduced me to the problem and advised me in preparing this paper. I would also like to thank Dr. Robert Sloan and Dr. Gyurgy Turan for introducing me to Professor Fortnow, and Steve Fenner for his useful comments on this paper. I am grateful to my past and current teachers, to MathPath and Canada/USA Mathcamp, and to the Illinois Mathematics and Science Academy's SIR program. Professor Fortnow is partially supported by NSF grants CCF-0829754 and DMS-0652521.

\bibliographystyle{amsplain}
\bibliography{PosetgamesareNC1hard}

\providecommand{\bysame}{\leavevmode\hbox to3em{\hrulefill}\thinspace}
\providecommand{\MR}{\relax\ifhmode\unskip\space\fi MR }
\providecommand{\MRhref}[2]{%
  \href{http://www.ams.org/mathscinet-getitem?mr=#1}{#2}
}
\providecommand{\href}[2]{#2}
\begin{thebibliography}{1}

\bibitem{bouton1901nim}
C.L. Bouton, \emph{{Nim, a game with a complete mathematical theory}}, Annals
  of Mathematics \textbf{3} (1901), no.~1, 35--39.

\bibitem{Byrnes03}
S.~Byrnes, \emph{Poset game periodicity}, INTEGERS: The Electronic Journal of
  Combinatorial Number Theory \textbf{3} (2003).

\bibitem{WinningWays}
J.~Conway E.~Berlekamp and R.~Guy, \emph{Winning ways for your mathematical
  plays}, Academic Press, New York, 1982.

\bibitem{Gale74}
D.~Gale, \emph{A curious nim-type game}, Amer. Math. Monthly \textbf{81}
  (1974), 876--879.

\bibitem{Grundy39}
P.M. Grundy, \emph{Mathematics and games}, Eureka (University of Cambridge
  Magazine) \textbf{2} (1939), 6--8.

\bibitem{Schuh52}
F.~Schuh, \emph{Spel van delers (game of divisors)}, Nieuw Tijdschrift voor
  Wiskunde \textbf{39} (2003), 299.

\bibitem{Soltys10}
M.~Soltys and C.~Wilson, \emph{On the complexity of computing winning
  strategies for finite poset games}, Theory of Computing Systems \textbf{48}
  (2010), no.~3, 680--692.

\bibitem{Sprague35}
R.P. Sprague, \emph{Uber mathematische kampfspiele}, Tohoku Mathematical
  Journal, First Series \textbf{41} (1935), 438--444.

\end{thebibliography}

\end{document}